\documentclass[12pt]{article}

\usepackage{amssymb,amsmath,amsthm,amscd,latexsym, lineno, longtable,color}

\usepackage{amsfonts}
\usepackage[mathscr]{eucal}
\usepackage{color}
\newtheorem{thm}{thm}[section]
\newtheorem{prop}[thm]{Proposition}
\newtheorem{lem}[thm]{Lemma}
\newtheorem{cor}[thm]{Corollary}
\newtheorem{theorem}[thm]{Theorem}
\newtheorem{ex}{Example}
\newtheorem{defn}{Definition}

\font\msbm=msbm10 at 12pt
\newcommand{\ZZ}{\mbox{\msbm Z}}

\newcommand{\FF}{\mbox{\msbm F}}
\newcommand{\F}{\mbox{\msbm F}}

\def\x{\mathbf{x}}

\def\y{\mathbf{y}}
\def\v{\mathbf{v}}

\def\vu{\mathbf{u}}
\def\vw{\mathbf{w}}
\def\vc{\mathbf{c}}

\def\1v{\mathbf{1}}
\def\0v{\mathbf{0}}

\begin{document}
%\begin{frontmatter}

\title{$\Theta_S-$cyclic codes over $A_k$}
\author{Irwansyah, Aleams Barra \\ {\it \small Institut Teknologi Bandung} \\ {\it \small Bandung, Indonesia} \\
Steven T. Dougherty \\
{\it \small University of Scranton} \\
{\it\small Scranton, PA, USA} \\
Ahmad Muchlis, Intan Muchtadi-Alamsyah \\ {\it\small Institut Teknologi Bandung} \\ {\it\small Bandung, Indonesia} \\
Patrick Sol\'e \\
{\it\small Telecom Paris Tech, Paris, France}\\
{\it\small and} \\ {\it\small University King Abdul Aziz, Jeddah, Saudi Arabia}\\
Djoko Suprijanto\footnote{Corresponding author. email: djoko@math.itb.ac.id} \\ {\it\small Institut Teknologi Bandung} \\ {\it\small Bandung, Indonesia} \\
Olfa Yemen \\ {\it\small Institut Pr\'{e}paratoire Aux Etudes D'Ing\'{e}nieurs} \\ {\it\small El Manar, Tunis, Tunisia} }

\maketitle

%\date{}

\begin{abstract}
We study $\Theta_S-$cyclic codes over the family of rings $A_k.$
We characterize $\Theta_S-$cyclic codes in terms of their binary images.
A family of Hermitian inner-products is defined and we prove that if a code
is $\Theta_S-$cyclic then its Hermitian dual is also $\Theta_S-$cyclic.
Finally, we give constructions of $\Theta_S-$cyclic codes.
\vspace{1cm}

\end{abstract}

%\end{frontmatter}
\section{Introduction}

Codes over rings are a widely studied object.  At the heart of this subject is
the construction of distance preserving maps from rings to the
 binary Hamming space.  Initially, the four rings of order 4, $\FF_4$,
 $\ZZ_4$, $\FF_2 + u \FF_2, u^2=0$ and $\FF_2 + v \FF_2, v^2=v$ were studied
with respect to their distance preserving maps.  See \cite{TypeIV},
for a complete description of codes over these rings.  The finite field $\FF_4$
and the ring $\ZZ_4$ are part of the well known families of finite fields and
integer modulo rings.  Codes over these rings have been well studied.
The ring $\FF_2 + u \FF_2$ has been generalized to the family of rings $R_k$.
See \cite{Dough1}, \cite{Dough2} and \cite{Dough2a} for a description of codes over
these rings.
Codes over the ring $\FF_2 + v \FF_2$  have been studied in numerous papers,
see  \cite{Nuh}, \cite{Dertli} and \cite{TypeIV} for example.
This family of rings has been generalized to the family of rings $A_k$ in \cite{C-D-D}.

Cyclic codes have long been one of the most interesting families
of codes because of their rich algebraic structure.  Namely, they can be
viewed as ideals in a polynomial ring. This connection allows for a classification of
cyclic codes by classifying ideals in a polynomial ring moded out by $x^n-1.$
In \cite{Boucher1}, skew cyclic codes were described as a generalization of cyclic codes.  This notion
was applied to codes over $\FF_2 + v \FF_2$ in \cite{Nuh}. In this work, they  described
generator polynomials of $\theta-$cyclic codes defined over this ring as well as the generator polynomials of their duals
with respect to both the Euclidean and the Hermitian inner product.  They also provided some examples of optimal
$\theta-$cyclic self-dual codes with respect to the Euclidean and Hermitian inner product.

Recently, in \cite{gao}, Gao further generalized previous work in  \cite{Nuh} and \cite{Boucher1} to
$\theta-$cyclic codes over the ring $\mathbb{F}_p+v\mathbb{F}_p,$ with $p$ prime.  Among his results
 is that $\theta-$cyclic codes over the ring $\mathbb{F}_p+v\mathbb{F}_p$ is equivalent
to either cyclic codes or quasi-cyclic codes over the same ring.  We note in \cite{Siap},  Siap et.al.
have proven the same result earlier for the case of $\theta-$cyclic codes over the field $\mathbb{F}_q.$

In this work, we study skew cyclic codes over the family of rings $A_k$ and we study
their images via a recursively defined  distance preserving map into  the binary Hamming space.

\section{Definitions and Notations}
\subsection{Family of Rings}

In \cite{C-D-D}, the following family of rings, which are a generalization of the ring $\FF_2+v\FF_2$, were defined.
 For the integer $k \geq 1$, let $$A_k = \FF_2[v_1,v_2,\dots,v_k] / \langle v_i^2 = v_i, v_iv_j = v_jv_i \rangle. $$
The rings in this family of rings are  finite commutative rings with cardinality $2^{2^k}$ and characteristic 2.

We shall describe a notation for any element in the ring $A_k$. Let $k$ be an integer with $k \geq 1.$
Let $B \subseteq  \{1,2,\dots,k \}$ and let $v_B = \prod_{i \in B} v_i.$
In particular, $v_{\emptyset} =1.$
It follows that  each element of $A_k$  is of the form
$\sum_{B \in {\cal P}_k} \alpha_B v_B$ where $\alpha_B \in \FF_2,$ and ${\cal P}_k$
is the power set of the set $ \{1,2, \dots, k  \}.$
For $A,B \subseteq  \{1,2,\dots,k \}$ we have that $v_Av_B = v_{A \cup B}$ which gives that
$$
\sum_{B \in {\cal P}_k } \alpha_B v_B  \cdot \sum_{C \in {\cal P}_k }  \beta_C v_C = \sum_{D \in {\cal P}_k}
\left( \sum_{B \cup C = D} \alpha_B \beta_C \right) v_D.
$$

It is shown, in \cite{C-D-D}, that    the  only unit in the ring $A_k$ is 1.
It is also shown that the
ideal $\langle w_1,w_2, \dots,w_k \rangle $, where $w_i  \in  \{  v_i  , 1+v_i \}$, is a maximal ideal of cardinality $2^{2^k-1}$.
Note that this gives $2^k$ maximal ideals.  Hence, except for the case when $k=0$,
namely the finite field of order 2, the ring is not a local ring.

The ring $A_k$ is a principal ideal ring.
In particular,
let $I = \langle \alpha_1,\alpha_2,\dots,\alpha_s \rangle$ be an ideal in $A_k$, then
$I$ is a principal ideal generated by the element which is the sum of all non-empty products of the $\alpha_i$, that is
$$
I = \left \langle  \sum_{\substack{ A \subseteq  \{1,2,\dots, s\}, \\
A \neq \emptyset } } \prod_{i \in A} \alpha_i       \right \rangle.
$$

 Let $S$ be a subset of $\{1,2,\dots, k \}$.   We shall define  a set of
 automorphisms in the ring $A_{k}$ based on the set $S$.  Define
the map $\Theta _{i}$ by
%$\ \Theta _{i}(0)=0$, $\Theta _{i}(1)=1$, $\ $%
\begin{equation*}
\Theta _{i}(v_{i})=v_{i}+1\hspace{1cm}\text{and}\hspace{1cm}
\Theta _{i}(v_{j})=v_{j},\hspace{0.5cm}\forall j\neq i.
\end{equation*}%

For all $\ S \subseteq  \{ 1, 2, \dots, k \}$ the
automorphism $\Theta _{S}$ is defined by:
\begin{equation*}
\Theta_{S}=\prod \limits_{i \in S}\Theta _{i}.
\end{equation*}
Note that $\Theta_S$ is an involution on the ring $A_k$.
We shall use this involution to define $\Theta_S-$cyclic codes.

In the ambient space $A_k^n$ we have two natural families of  inner-products.
First we have the standard  Euclidean inner-product:
$ [ \vw,\vu] = \sum w_j u_j .$

Now we define the Hermitian inner-product with respect  to a subset $T \subseteq \{1,2, \dots, k \}.$
%Note that  we  use $T$ as this subset since we use $S$ in terms of the $\Theta_S$ cyclic codes to be defined later.
Define
$$ [ \vw,\vu]_{H_T} = \sum w_j  \Theta_T ({u_j}). $$

We have orthogonals corresponding to each inner-product.
We define $C^\perp =  \{ \vw \ | \ [\vw,\vu]=0, \ \forall \vu \in C\}$
and
$C^{H_T} =  \{ \vw \ | \ [\vw,\vu]_{H_T}=0, \ \forall \vu \in C\}.$
Note that there are $2^k$ possible Hermitian duals for codes over $A_k$.    Moreover, we notice that if $T= \emptyset$
the Hermitian orthogonal is in fact the Euclidean inner-product.  We retain the Euclidean orthogonal notation because of
 its importance as a particular example of orthogonals.

  Since the ring $A_k$  is a Frobenius ring,  in both cases we have the standard cardinality condition, namely
$|C| |C^\perp| = |A_k^n|$ and $|C| |C^{H_T}| = |A_k^n|.$ Furthermore,
we also have several properties of the ring $A_k$ as follows.  See  \cite{C-D-D}, for  foundational material on these rings.

\begin{prop} \label{prop2.1}
If $I$ is a maximal ideal in $A_k$, then $I=\langle w_1,w_2,\dots,w_k\rangle$, where $w_i\in\{v_i,v_i+1\},$
for $1 \leq i \leq k.$
\end{prop}
\begin{proof}
Since $A_k$ is a principal ideal ring, we have $I=\langle \omega\rangle$, for some $\omega\in A_k$. Let
\[
\omega=\sum_{B\in {\cal P}_k} \alpha_B v_B
\]
for some $\alpha_B\in\mathbb{F}_2,$ $v_B=\prod_{i\in B}v_i,$ and $v_\emptyset = 1.$

If $\alpha_{\emptyset}=0$, then clearly $\omega\in \langle v_1,v_2,\dots,v_k\rangle$.
On the other hand, if $\alpha_{\emptyset}=1,$ then let
\[
\omega=1+v_{B'}+\sum_{B\in {\cal P}_k\backslash \{B^\prime, \emptyset\}}\alpha_B v_B
\]
where $B^\prime \in {\cal P}_k$ such that $\alpha_{B^\prime}=1.$
We use mathematical induction on the cardinality of $B^\prime$.
If $B^\prime=\{j_1\}$, then $\omega \in \langle v_{j_1}+1,v_{j_2},\dots,v_{j_k} \rangle.$
 If $B^\prime=\{j_1,j_2\},$ then consider
\begin{align*}
v_{B^\prime}+1 &= (v_{j_1}+1)(v_{B^\prime}+1)+v_{j_1}(v_{B^\prime \backslash \{j_1\}}+1)\\
               &= (v_{j_1}+1)(v_{j_1}v_{j_2}+1)+v_{j_1}(v_{j_2}+1).
\end{align*}
As a consequence, $\omega \in \langle v_{j_1}+1,v_{j_2}+1,v_{j_3},\ldots,v_{j_k} \rangle.$
Moreover, if $B^\prime=\{j_1,j_2,j_3\},$ then
\[
v_{B^\prime}+1=(v_{j_1}+1)(v_{j_1}v_{j_2}v_{j_3}+1)+v_{j_1}(v_{j_2}v_{j_3}+1).
\]
By the previous equation we have
\[
v_{B^\prime}+1=(v_{j_1}+1)(v_{j_1}v_{j_2}v_{j_3}+1)+v_{j_1}\left((v_{j_2}+1)(v_{j_2}v_{j_3}+1)+v_{j_2}(v_{j_3}+1)\right),
\]
consequently, $\omega\in \langle v_{j_1}+1,v_{j_2}+1,v_{j_3}+1, v_{j_4},\ldots,v_{j_k}\rangle.$
Now, assume that the above equation is
satisfied when $|B^\prime|=m-1$, then when $|B^\prime|=m,$ we have
\[
v_{B^\prime}+1=(v_{j_m}+1)(v_{B^\prime}+1)+v_{j_m}(v_{B^\prime \backslash \{j_m\}}+1).
\]
This implies that $\omega\in J=\langle v_{j_1}+1,\dots,v_{j_m}+1,v_{j_{m+1}},\dots,v_{j_k}\rangle.$
Since $J$ and $I$ are maximal ideals, and  $I\subseteq J,$ then we have that $I=J.$
\end{proof}
\begin{lem}
 The ring $A_k$ can be viewed as an $\mathbb{F}_2$-vector space with dimension $2^k$ whose basis consists of elements of the form
$\prod_{i\in B}w_i$, where $B \in {\cal P}_k$ and $w_i\in\{v_i,v_i+1\}.$
\label{lemavector}
\end{lem}
\begin{proof}
 Every element $a\in A_k$ can be written as $a=\sum_{B\in {\cal P}_k}\alpha_Bv_B$, for some $\alpha_B\in\mathbb{F}_2$,
 $v_B=\prod_{i\in B} v_i$ and $v_\emptyset=1$. Therefore, $A_k$ is a vector space over $\mathbb{F}_2$ with a basis that consists
of elements of the form $v_B=\prod_{i\in B} v_i$, where $v_\emptyset=1.$ There are $\sum_{j=0}^k{k\choose j}=2^k$
basis elements. Since all $v_B$
are linearly independent over $\mathbb{F}_2$,
so is the element $w_B=\prod_{i\in B}w_i$, where $w_i\in\{v_i,1+v_i\}$.
\end{proof}

As a direct consequence of the above proposition and \cite[Theorem 2.3]{C-D-D}, we have the following theorem.
\begin{theorem}
 An ideal $I$ in $A_k$ is maximal if and only if $I=\langle w_1,w_2,\dots,w_k\rangle.$
\label{teoideal}
\end{theorem}
The following proposition gives a characterization of automorphisms in $A_k.$
\begin{prop}
 Let $\theta$ be an endomorphism in $A_k$. The map $\theta$ is an automorphism
if and only if $\theta(v_i)=w_j$, where $w_j\in\{v_j,v_j+1\}$, for every $i\in \{1,2,\ldots,k\},$
$\theta(v_i)\not= \theta(v_j)$ when $i\not=j,$
and $\theta(a)=a,$ for every $a\in\mathbb{F}_2.$
\label{auto}
\end{prop}
\begin{proof}
 Let $J=\langle v_1,\dots,v_k\rangle$ and $J_{\theta}=\langle \theta(v_1),\dots,\theta(v_k)\rangle.$
Consider the map
\[
\begin{array}{llll}
\lambda : & \displaystyle{\frac{A_k}{J}} & \longrightarrow & \displaystyle{\frac{A_k}{J_{\theta}}}\\
	   & t+J & \longmapsto & \theta(t)+J_{\theta}. \\
\end{array}
\]
We can see that the map $\lambda$ is a ring homomorphism.
For any $a,b\in A_k/J$ where $\lambda(a)=\lambda(b)$, let $a=a_1+J$ and $b=b_1+J$ for some $a_1,b_1\in A_k.$
We have that
$\theta(a_1-b_1)\in J_{\theta},$ so $a_1-b_1\in J.$ Consequently, $a-b=0+J,$ which means $a=b,$
in other words, $\lambda$ is a monomorphism.
Moreover, for any $a^\prime  \in A_k/J_{\theta},$ let $a^\prime=a_2+J_{\theta}$ for some $a_2\in A_k,$
then there exists $a=\theta^{-1}(a_2)+J$ such that
$\lambda(a)=a^\prime.$ Therefore, $\mathbb{F}_2\simeq A_k/J\simeq A_k/J_{\theta},$
which implies $J_{\theta}$ is also a maximal ideal.
By Theorem \ref{teoideal}, $J_{\theta}=\langle w_1,w_2,\ldots,w_k\rangle.$  By \cite[Theorem 2.6]{C-D-D},
\[
J_\theta=
\left\langle \sum_{B \in {\cal P}_k \backslash\{\emptyset\}}\prod_{i\in B}w_i \right\rangle
=\left \langle \sum_{B \in {\cal P}_k \backslash \{ \emptyset\}}
\prod_{i\in B}\theta(v_i) \right\rangle
\]
which means, $\sum_{B \in {\cal P}_k \backslash \{ \emptyset\}} \prod_{i\in B}w_i$ and
$\sum_{B \in {\cal P}_k \backslash \{ \emptyset\}}
\prod_{i\in B}\theta(v_i)$ are associate. Since the only unit in $A_k$ is $1,$ then they must be equal.
 Therefore, $\theta(v_{i})=w_j.$  Since $\theta$ is an automorphism, $\theta(v_i)\not=\theta(v_j)$ whenever $i\not=j.$

Suppose that $\theta(v_i)=w_j,$ and $\theta(v_i)\not=\theta(v_j)$ whenever $i\not=j.$
 By Lemma \ref{lemavector}, we can see that $\theta$ is also an
automorphism.
\end{proof}
We have the following immediate consequence.

\begin{cor}
An involution $\psi$  of   $A_k$ must be $\Theta_S$ or it  must  have
$\psi(v_i) = w_j \in \{v_j,1+v_j\} $ and $\psi(w_j) = v_i$
 where $i,j \in \{1,\dots,k\}$.
\end{cor}
\begin{proof}
By the previous proof we have that the only automorphisms send $v_i$ to either $v_j$ or $1+v_j$.
Then to be an involution the application of the map twice must be the identity which gives the result.
\end{proof}

\subsection{Gray Map}

Gray maps were defined on all commutative rings of order 4, see \cite{TypeIV} for a description of these four Gray maps.
For the ring $A_1 = \FF_2 + v \FF_2$,
we have  the Gray map $\phi_1:A_1 \rightarrow \FF_2^2 $ defined by $\phi(a + b v_1) = (a,a+b)$.
For $A_1$ this is realized as
\begin{eqnarray*}
0 &\rightarrow 00 \\
1 &\rightarrow 11 \\
v &\rightarrow 01 \\
1+v &\rightarrow 10. \\
\end{eqnarray*}
We extend this map inductively as follows.  Every element in the ring $A_k$
can be written as $\alpha + \beta v_k$, where $\alpha, \beta \in A_{k-1}.$
Then for $k \geq 2$, define $\phi_k:A_k \rightarrow A_{k-1}^2$ by
\[
\phi_k(\alpha + \beta v_k) = (\alpha, \alpha + \beta).
\]
Then define $\Phi_k : A_k \rightarrow \FF_2^{2^k}$ by
$\Phi_1(\gamma) = \phi_1(\gamma)$,
$\Phi_2(\gamma) = \phi_1(\phi_2(\gamma))$ and
\[
\Phi_k (\gamma) = \phi_1 (\phi_2(  \dots (\phi_{k-2} ( \phi_{k-1} ( \phi_k (\gamma)) \dots ).
\]
It follows immediately that $\Phi_k(1) = {\bf 1 },$ the all one vector.

We note that the Gray map is a bijection and is a linear map.

\begin{ex}
For $k=2$, we have that  $\Phi_2: A_2\rightarrow {\mathbb{F}}_2^{4}$ and
\begin{equation*}
\Phi_2(a +bv_1 + c v_2  + d v_1v_2) =
(a, a+b, a+c, a+b+c+d).
\end{equation*}
\end{ex}

Let $T$ be the matrix that performs the cyclic shift on a vector.  That is $T(v_1,v_2,\dots, v_n) = (v_n,v_1,\dots, v_{n-1}).$
Let $\sigma _{i,k}$ be the permutation on  $ \{ 1, 2, \dots, 2^{k} \}$ defined by
\begin{equation*}
(\sigma _{i,k})_{\{ p2^{i}+1, \dots, (p+1)2^{i} \}}= T^{2^{i-1}}( p2^{i}+1, \dots, (p+1)2^{i}  ),
\end{equation*}
for all $0 \leq p \leq 2^{k-i}-1$.
Let $\Sigma _{i,k}$ be the permutation on
elements of ${\mathbb{F}}_{2}^{2^{k}}$ induced by $\sigma _{i,k}$. That is,
for ${\mathbf{x}}=(x_{1},x_{2},\ldots ,x_{2^{k}})\in {\mathbb{F}}%
_{2}^{2^{k}} $,
\begin{equation}
\Sigma _{i,k}({\mathbf{x}})=\left( x_{\sigma _{i,k} (1)},x_{\sigma
_{i,k}(2)},\ldots ,x_{\sigma _{i,k} (2^{k})}\right) \text{.}
\label{eq:induce}
\end{equation}

In other word, $\Sigma_{i,k}$ is a permutation induced by $\sigma_{i,k}.$
To clarify the permutation above, let us consider an example below.
\begin{ex}
Let $k=2$ and $i=1.$ We want to find $\Sigma_{1,2}(\x),$ for all $\x=(x_1,x_2,x_3,x_4)\in\mathbb{F}_2^{2^2}.$
Since $i=1,$ we have $0\leq p\leq 1.$ Hence, for $p=0$ we have
\[
\left(\sigma_{1,2}\right)_{\{1,2\}}=T(1,2)=(2,1)
\]
and for $p=1$ we have
\[
\left(\sigma_{1,2}\right)_{\{3,4\}}=T(3,4)=(4,3).
\]
Therefore,
\[
\begin{aligned}
\Sigma_{1,2}(\x) &=  \left(x_{\sigma_{1,2}(1)},x_{\sigma_{1,2}(2)},x_{\sigma_{1,2}(3)},x_{\sigma_{1,2}(4)}\right)\\
                 &=  \left(x_2,x_1,x_4,x_3\right).
\end{aligned}
\]

For $i=2,$ we have $p=0.$ Then,
\[
\left(\sigma_{2,2}\right)_{\{1,2,3,4\}}=T^{2}(1,2,3,4)=(3,4,1,2).
\]
Therefore, we have
\[
\begin{aligned}
\Sigma_{2,2}(\x) & =  \left(x_{\sigma_{2,2}(1)},x_{\sigma_{2,2}(2)},x_{\sigma_{2,2}(3)},x_{\sigma_{2,2}(4)}\right)\\
                 & =  \left(x_3,x_4,x_1,x_2\right).
\end{aligned}
\]
\label{sigma}
\end{ex}

We can also define another map which will be used later for constructing generators
for $\Theta_S-$cyclic codes and optimal codes.
Let $p,k\in\mathbb{N}$, where $p< k$. Let
$\Omega_p=\{p+1,\dots,k\}$ and $s=2^{k-p}$. We have that $|{\cal P}_k(\Omega_p)|=s$.
We can define a Gray map as follows:
\[
\Psi_{k,p}: A_k\rightarrow A_p^s.
\]
Denote the coordinates of $A_p^s$ by the lexicographic ordering of the subsets
of $\Omega_p$ and denote them by $B_1,\dots,B_s$. Note that $B_1=\emptyset$
and $B_s=\Omega_p$. An element of $A_k$ can be written as $\sum_{B\subseteq\Omega_p}\alpha_Bw_B$,
where $\alpha_B\in A_p$ and $w_B=\prod_{i \in B}v_i$. Then
\[
\Psi_{k,p}\left(\sum_{B\subseteq\Omega_p}\alpha_Bw_B\right)=
\left(\sum_{D\subseteq B_1}\alpha_D,\sum_{D\subseteq B_2}\alpha_D,\dots,
\sum_{D\subseteq B_s}\alpha_D\right).
\]
For $p=0$, $\Psi_{k,0}$ is the same map as $\Psi_k$ \cite{C-D-D}.
In that paper, it is shown that $\Psi_k$ and $\Phi_k$ are conjugate, {\it i.e.}
their images are permutation equivalent.

Notice that the representation of an element in $A_k$
can be changed by replacing any $v_i$ with $1+v_i$. In that way, we can let $u_i$ be either $v_i$
or $v_i+1$ for each $i$. Then we have an alternative definition of $\Psi_{k,p}$ as
\[
\overline{\Psi}_{k,p}\left(\sum_{B\subseteq\Omega_p}\alpha_By_B\right)
=\left(\sum_{D\subseteq B_1}\alpha_D,\sum_{D\subseteq B_2}\alpha_D,\dots,
\sum_{D\subseteq B_s}\alpha_D\right)
\]
where $y_B=\prod_{i\in B}u_i$. Any result for $\Psi_{k,p}$ can be replaced for $\overline{\Psi}_{k,p}$.

\begin{lem}
Let  $k \geq 1$ and   $1 \leq i \leq k$.  For $x\in A_{k}$ we have
\begin{equation*}
\Sigma _{i,k}(\Phi _{k}(x))=\Phi _{k}(\Theta_i (x)).
\end{equation*}
\end{lem}
\begin{proof}
Let $\ 1\leq i\leq k$. For $k=1 $, we have $x= a+b v$ and $\Theta_1(x)= (
a+b) +b v$. Thus $\Phi_1(\Theta_1(x)) = ( a+b,a) = \Sigma_{1,1} ( \Phi_1(x)).$

Assume the result is true for values less than or equal to $k-1$ and let $i < k $ with
$x\in A_{k}$.   We have $x = a+b v_k $ with $a,b \in A_{k-1}$ and
$\Theta_i(x)= \Theta_i(a)+ \Theta_i(b)v_k$ . Then
\begin{equation*}
\phi_k(x) = (a, a+b).
\end{equation*}
It follows that
\begin{equation*}
\Phi_k(x) = ( \Phi_{k-1}(a),\Phi_{k-1}(a+b)) \hspace{0.5cm} \text{and} \hspace{0.5cm}
\Phi_k(\Theta_i(x)) = ( \Phi_{k-1}(\Theta_i(a)), \Phi_{k-1}(\Theta_i(a)+ \Theta_i(b))).
\end{equation*}
We note that
\begin{equation*}
\Phi_{k-1}(a) =( x_1,x_2, \ldots,x_{2^{k-1}}),\hspace{1cm}
\Phi_{k-1}(a+b) =(x^\prime_1, x^\prime_2, \ldots,x^\prime_{2^{k-1}})
\end{equation*}
and
\begin{equation*}
\Phi_k(x) = ( \Phi_{k-1}(a),\Phi_{k-1}(a+b)) = ( y_1, y_2, \ldots,y_{2^{k-1}},
y_{2^{k-1}+1}, \ldots,y_{2^{k}})
\end{equation*}
with $( x_1,x_2, \ldots,x_{2^{k-1}}) = ( y_1,y_2,\ldots,y_{2^{k-1}})$ and $( x^\prime_1,
x_2,\ldots,x^\prime_{2^{k-1}}) = ( y_{2^{k-1}+1}, y_{2^{k-1}+2},\ldots,y_{2^{k}})$.

We have that
\[
\begin{aligned}
\Sigma_{i,k}(\Phi_k(x)) %&= \Sigma_{i,k}( ( y_1, \ldots,y_{2^{i}},y_{2^{i}+1},\ldots,
%y_{2.2^{ i}},\ldots,y_{2^{k-1}}, y_{2^{k-1}+1}, \ldots,y_{2^{k}})) \\
&= \Sigma_{i,k}( (\underbrace{y_1,\ldots,y_{2^{i}},}_{2^i} \underbrace{y_{2^{i}+1},\ldots, y_{2.2^{
i}},}_{2^i}\ldots, \underbrace{y_{(2^{k-i}-1)2^i+1}, \ldots,y_{2^k}}_{2^i})) \\
&=( \underbrace{y_{2^{i-1}+1}, y_{2^{i-1}+2}, \ldots, y_{2^i}, y_1,\ldots, y_{2^{i-1}},}_{2^i}\\
&\qquad  \underbrace{y_{2^{i+1}-2^{i-1}+1},y_{2^{i+1}-2^{i-1}+2},\ldots,y_{2^i+1}, \ldots,y_{2^{i+1}-2^{i-1}}}_{2^i},\\
&\qquad  \ldots, \underbrace{y_{2^{i-1} \cdot (2^{k-i+1}-1)+1},y_{2^{i-1} \cdot (2^{k-i+1}-1)+2},\ldots,
y_{2^k}, \ldots,y_{2^{i-1} \cdot (2^{k-i+1}-1)}}_{2^i} ) \\
&= ( \Sigma_{i,k-1}(\Phi_{k-1}(a), \Sigma_{i,k-1}(\Phi_{k-1}(a+b)).
\end{aligned}
\]
The hypothesis of recurrence implies that
\[
\begin{aligned}
\Sigma_{i,k-1}(\Phi_k(x)) &= ( \Phi_{k-1}(\Theta_i(a)),
\Phi_{k-1}(\Theta_i(b))) \\
&= \Phi_k(\Theta_i(x)).
\end{aligned}
\]

Now  consider the case when $i=k$.  For $x\in A_{k}$,  $x = a+b v_k $
with $a,b \in A_{k-1}$ we have that  $\Theta_k(x)= (a+b) + b v_k$ . Then
\begin{equation*}
\phi_k (x) = (a, a+b),\hspace{1cm} \phi_k (\Theta_k(x)) = ( a+b,a)
\end{equation*}
and consequently
\begin{equation*}
\Phi_k(x) = ( \Phi_{k-1}(a),\Phi_{k-1}(a+b))  \quad \text{ and } \quad
\Phi_k(\Theta_k(x)) = ( \Phi_{k-1}(a+b ),\Phi_{k-1}(a+b)).
\end{equation*}
We note that
\begin{equation*}
\Phi_{k-1}(a) =( x_1,\ldots, x_{2^{k-1}}), \hspace{1cm} \Phi_{k-1}(a+b) =(
x^{\prime }_1, \ldots, x^{\prime }_{2^{k-1}})
\end{equation*}
and
\begin{equation*}
\Phi_k(x) = ( \Phi_{k-1}(a),\Phi_{k-1}(a+b)) = ( y_1,y_2,\ldots, y_{2^{k-1}},
y_{2^{k-1}+1}, \ldots, y_{2^{k}})
\end{equation*}
with $(x_1, \ldots,x_{2^{k-1}}) = ( y_1,\ldots,y_{2^{k-1}})$ and $( x^{\prime }_1,
\ldots, x^{\prime }_{2^{k-1}}) = ( y_{2^{k-1}+1},\ldots,y_{2^{k}})$.\newline
On the other hand,   when $i=k$ we have that
\[
\Sigma_{k,k}= T ^{2^{k-1}}.
\]
Then we have that
\[
\begin{aligned}
\Sigma_{k,k}(\Phi_k(x)) &= \Sigma_{i,k}(( y_1, \ldots,y_{2^{k-1}},
y_{2^{k-1}+1}, \ldots,y_{2^{k}})) \\
&= T ^{2^{k-1}} (( y_1, \ldots,y_{2^{k-1}}, y_{2^{k-1}+1}, \ldots,y_{2^{k}})) \\
&= ( y_{2^{k-1}+1}, \ldots,y_{2^{k}},y_1, \ldots,y_{2^{k-1}}) \\
&=( \Phi_{k-1}(a+b), \Phi_{k-1}(a)) \\
&= \Phi_k(\Theta_k(x)).
\end{aligned}
\]
\end{proof}

{
\begin{ex}
Consider
\[
\begin{aligned}
\Phi_2(v_1v_2)   & =  (\Phi_1(0),\Phi_1(v_1))\\
                 & =  (0,0,0,1).
\end{aligned}
\]
Then
\[
\begin{aligned}
\Phi_2(\Theta_1(v_1v_2)) & = \Phi_2((v_1+1)v_2)\\
                         & =  (\Phi_1(0),\Phi_1(1+v_1))\\
                         & =  (0,0,1,0)\\
                         & =  \Sigma_{1,2}(0,0,0,1)\\
                         & = \Sigma_{1,2} (\Phi_2(v_1v_2)),
\end{aligned}
\]
also,
\[
\begin{aligned}
\Phi_2(\Theta_2(v_1v_2)) & = \Phi_2(v_1+v_1v_2)\\
                         & = (\Phi_1(v_1),\Phi_1(0))\\
                         & = (0,1,0,0)\\
                         & = \Sigma_{2,2}(0,0,0,1)\\
                         & = \Sigma_{2,2} (\Phi_2(v_1v_2)).
\end{aligned}
\]
\end{ex}
}

We can extend the definition of $\Sigma$ to subsets of $\{1,2,\dots, k \}.$

\begin{defn}
For all $\ A\subseteq \{ 1, 2, \dots, k \}$ we define the
permutation $\Sigma_{A,k} $ by
\begin{equation*}
\Sigma _{A,k}=\prod\limits_{i\in A}\Sigma _{i,k}.
\end{equation*}
\end{defn}

It is clear that  for all $x\in A_{k}$ we have
\begin{equation}
\Sigma _{A,k}(\Phi _{k}(x))=\Phi _{k}(\Theta_{A} (x)).
\label{olfa}
\end{equation}

\section{ $\Theta_S -$cyclic codes over $A_{k}$}

We can now define skew cyclic codes using this family of rings and family of automorphisms.

\begin{defn}
A subset $C$ of $A_k^n$ is called a skew cyclic code of length $n $ if and
only if
\begin{itemize}
\item  $C$ is a submodule of $A_k^n$,
and
\item $C$ is invariant under the $\Theta_S-$shift $T_{\Theta_S}.$
\end{itemize}
\end{defn}

The second condition can be seen as follows.
If
\[
\vc = (c_0, c_1, \ldots, c_{n-1}) \in C
\]
then
\[
T_{\Theta_S}(\vc) =(\Theta_S(c_{n-1}), \Theta_S (c_0), \ldots, \Theta_S
(c_{n-2}))\in C.
\]

We can now describe the ambient algebraic space we shall use to describe skew cyclic codes.

\begin{defn}
Let $S \subseteq \{ 1, 2, \dots, k \} $. Define the skew
polynomial ring $A_k[x, \Theta_S ] =  \{a_0 + a_1 x + a_2x^2 + \cdots+a_n x^n|~
a_i \in A_k \}$ as the set of polynomials over $A_k$ where the addition
is the usual polynomial addition and the multiplication is defined by the
basic rule
\begin{equation*}
xa =\Theta_S (a) x
\end{equation*}
and extended to all elements of $A_k[x, \Theta_S ]$ by associativity and
distributivity.
\end{defn}

It is clear that the set $(A_k)_n = A_k[x, \Theta_S] / { \langle x^n -1 \rangle }$ is a left $
A_k[x, \Theta_S]$-module.
The next theorem follows immediately from the definitions.

\begin{theorem}
A code $C $ in $A_k^n$ is a $\Theta_S-$cyclic code if and only if $C$ is a
left $A_k[x, \Theta_S]$-submodule of the left $A_k[x, \Theta_S]$-module $
(A_k)_n.$
\end{theorem}

The following theorem characterizes  skew-polynomial rings over $A_k$.

\begin{theorem}
Let $A_k[x,\Theta_{S_1}]$ and $A_k[x,\Theta_{S_2}]$ be two skew polynomial rings over $A_k$.
Then, $A_k[x,\Theta_{S_1}]\cong A_k[x,\Theta_{S_2}]$ if and only if $|S_1|=|S_2|$.
\end{theorem}
\begin{proof}
Let $\lambda$ be an automorphism on $A_k$. By Theorem 3 in \cite{rim} and properties in \cite{C-D-D}, we have
$\lambda : x\longmapsto b_0+b_1x+\dots+b_nx^n$ can be extended to
an isomorphism between $A_k[x,\Theta_{S_1}]$ and $A_k[x,\Theta_{S_2}]$ if and only if all the following conditions hold:
\begin{itemize}
 \item $\lambda(\Theta_{S_1}(a))b_i=b_i\Theta_{S_2}^{i}(\lambda(a))$ for every $a\in A_k$ and for all $i=0,1,\dots,n$,
\item $b_1=1$,
\item $b_i=0$ for all $i=2,3,\dots,n$.
\end{itemize}
This implies that  $x\longmapsto x+b_0$. It follows that  we only need to check the following conditions:
\[
\lambda(\Theta_{S_1}(a))b_0=b_0 \lambda(a)
\]
and
\[
\lambda(\Theta_{S_1}(a))=\Theta_{S_2}(\lambda(a)).
\]
For the first condition, if we choose $a=v_i$ where $i\in S_1$,
then we have $\lambda((1+v_i)-v_i)b_0=0$, which gives $b_0=0$. Moreover,
the second condition gives $\Theta_{S_2}=\lambda\Theta_{S_1}\lambda^{-1}$,
and consequently, $\Theta_{S_1}$ and $\Theta_{S_2}$ are conjugate to each other.
Which means, they have to has the same cycle structures as an automorphism.
This gives $|S_1|=|S_2|$, and the result follows.
\end{proof}

As an immediate consequence, we have the following result.

\begin{cor}  Let $S \subseteq \{1, 2, \ldots, k \}.$
There are $k\choose |S|$ skew-polynomial rings $A_k[x,\Theta_{S'}]$ which are isomorphic to a given $A_k[x,\Theta_S]$.
\end{cor}
\begin{proof}
It is immediate since the number of subsets of $\{1,\dots,k\}$ with cardinality $|S|$ is $k\choose |S|$.
\end{proof}

\section{Characterizations of $\Theta_S-$cyclic codes}

\subsection{First characterization}

Let $S \subseteq \{ 1, 2, \ldots, k \}$ and let $\Sigma_S
= \tau_S \circ T^{2^k}$ be the permutation on elements of ${\mathbb{F}}%
_2^{n2^k}$ where $T $ is the cyclic shift modulo $n2^k$ and $\tau_S$ is the
permutation on elements of ${\mathbb{F}}_2^{n2^k}$ defined for all
\[
\x=(x_1^{1},\ldots, x_{_{2^k}}^{1},x_1^{2},\ldots,
x_{_{2^k}}^{2},\ldots,x_1^{n},\ldots, x_{2^k}^{n}) \in {\mathbb{F}}_2^{n2^k,}
\] by
\[
\begin{aligned}
%\begin{eqnarray*}  \label{eq:induce}
\tau_S({\mathbf{x}}) &= \tau_S((x_1^{1},\ldots, x_{2^k}^{1},x_1^{2},\ldots,
x_{2^k}^{2}, \ldots,x_1^{n},\ldots, x_{2^k}^{n})) \\
&= ( \Sigma_{S,k}(x^{1}),\Sigma_{S,k}(x^{2}),
\ldots, \Sigma_{S,k}(x^{n}) )
%\end{eqnarray*}
\end{aligned}
\]
where $x^j = (x_1^{j}, \ldots, x_{2^k}^{j})$. Since $T^{2^k}$ and $
\tau_S$ commute, $\Sigma_S$  can be written as $T^{2^k} \circ \tau_S$ as well.
 Let $\sigma_S$ denote permutation on $\{1,2,\ldots,n2^k\},$ the indices of elements in
${\mathbb{F}}_2^{n2^k},$ that induce
the permutation $\Sigma_S$ above. 

\begin{lem}
\label{lem0}
Let $C$ be a code in $A_k^n$.
The code $\Phi_k(C)$ is fixed by the permutation $\Sigma_S $
if and only if $C$ is a $\Theta_S-$cyclic code.
\end{lem}
\begin{proof}
Let $C$ be a $\Theta_S-$cyclic code and let $ \y= (y_1,y_2,\ldots,y_{n2^k}) \in \Phi_k(C)$.
That is, there exists  $ \x=(x_1,x_2,\ldots,x_{n})\in C $
such that     $ \y=  (\Phi_k(x_1),\Phi_k(x_2),\ldots,\Phi_k(x_n))$
and $\Phi_k(x_i)\in {\mathbb{F}}_2^{2^k}$.

We have that
\[
\begin{aligned}
\Sigma_S(\y) & = \tau_S \circ T^{2^{k}}(  (\Phi_k(x_1),\Phi_k(x_2),\ldots,\Phi_k(x_n)))\\
      & =  \tau_S (  (\Phi_k(x_n),\Phi_k(x_1),\ldots,\Phi_k(x_{n-1})))\\
  &= (\Sigma_{S,k}(\Phi_k(x_n)),  \Sigma_{S,k}
  (\Phi_k(x_1)), \ldots, \Sigma_{S,k}\Phi_k(x_{n-1}))\\
  &= \Phi_k(\Theta_S (x_n)),  \Phi_k(\Theta_S (x_1)),\ldots,\Phi_k(\Theta_S (x_{n-1}
))) \\
   &= \Phi_k(\Theta_S (\x))\in \Phi_k( C),
\end{aligned}
\]
since $C$ is  a $\Theta_S-$cyclic code.

Let   $C^\prime$ be a code  in $\F_2^{n2^k}$
that is fixed  by  the  permutation $ \Sigma_S$
and let $ \x= (x_1,x_2,\ldots,x_{n})
\in \Phi_k^{-1} (C^\prime)$.  We have that  there exists
   $ \y= ( x_{1}^{1},\ldots,
x_{1}^{2^k}, x_{2}^{1}, \ldots, x_{1}^{2^k}, \ldots, x_{n}^{1}, \ldots,
x_{n}^{2^k})\in C^\prime$, such that
$$
\y= \Phi_{k}((x_1,x_2,\ldots,x_{n})) =
(\Phi_k(x_1),\Phi_k(x_2), \ldots,\Phi_k(x_n)).
$$

Since $C^\prime$ is fixed by  the  permutation $ \Sigma_S $ and
$\Phi_k(x_i)\in {\mathbb{F}}_2^{2^k}$, we have that for $\Sigma_S(\y) \in C^\prime$:
\[
\begin{aligned}
%\begin{eqnarray*}
\Sigma_S(\y) & =  \tau_S \circ T^{2^{k}}(
(\Phi_k(x_1),\Phi_k(x_2),\ldots,\Phi_k(x_n)))\\
  &= \tau_S (
(\Phi_k(x_n),\Phi_k(x_1),\ldots,\Phi_k(x_{n-1})))\\
  &=
(\Sigma_{S,k}(\Phi_k(x_n)),  \Sigma_{S,k}(\Phi_k(x_1)),\ldots,
\Sigma_{S,k}\Phi_k(x_{n-1}))\\
   &= (\Phi_k(\Theta_{S}(x_n)),
\Phi_k(\Theta_{S}(x_1)),\ldots,\Phi_k(\Theta_{S}(x_{n-1}))\\
 &= \Phi_k(\Theta_{S}(\x)).
%\end{eqnarray*}
\end{aligned}
\]
It follows that $
\Phi_k(\Theta_{S}(\x)) \in C^\prime$ and thus
  $\Phi_{k}^{-1}(C^\prime)$ is a $\Theta_S-$cyclic code .
\end{proof}

Next, we consider the order of the permutation $\sigma_S$ which we will need in Theorem~\ref{oioi}.

\begin{lem} \label{2and2'}
\label{lem1} Let $ord(\sigma_S)$ denote the order of $\sigma_S.$
Then
\begin{equation}
ord(\sigma_S) = \left\{ \begin{array}{cc}
n & n \equiv 0 \pmod{2} \\
2n & n \equiv 1 \pmod{2}  \\
\end{array}. \right.
\end{equation}
\end{lem}
\begin{proof}
It is clear that  for $ n$ odd, the permutation $ \sigma_S $  is  a product of $2^{k-1}$  cycles of length $2n$  and we have
\[
\begin{aligned}
%\begin{eqnarray*}\label{2}
\sigma_S &=  ( 1, \sigma_S(1), \sigma_S^2(1),\ldots,\sigma_S^{2n-1}(1)) \\
& \quad ( 2, \sigma_S(2), \sigma_S^2(2),\ldots,\sigma_S^{2n-1}(2)) \cdots \\
& \quad ( 2^{k-1}, \sigma_S(2^{k-1}),
\sigma_S^2(2^{k-1}),\ldots,\sigma_S^{2n-1}(2^{k-1})).
%\end{eqnarray*}
\end{aligned}
\]
Furthermore,   if $ n$ is even, we have  $ \sigma_S $  is  a product of $2^{k}$  cycles of length $ n$  and
\[
\begin{aligned}
%\begin{eqnarray*}\label{2'}
\sigma_S &= ( 1, \sigma_S(1), \sigma_S^2(1),\ldots,\sigma_S^{n-1}(1)) \\
&\quad ( 2, \sigma_S(2), \sigma_S^2(2), \ldots,\sigma_S^{n-1}(2)) \cdots \\
&\quad ( 2^{{k}}, \sigma_S(2^{k}), \sigma_S^2(2^{k}),\ldots,\sigma_S^{n-1}(2^{k})).
%\end{eqnarray*}
\end{aligned}
\]
\end{proof}

\begin{theorem} \label{oioi}
\label{th1} Let $C$ be a code in $A_k^{n}$.

\begin{enumerate}
\item If $n$ is odd then $C$ is a skew-cyclic code if and only if $\Phi_k(C)$
is equivalent to an additive $2^{k-1}-$quasi-cyclic code $C^{\prime }$ in $
{\mathbb{F}}_{2}^{n2^k}$.

\item If $n$ is even then $C$ is a skew-cyclic code if and only if $\Phi_k(C)$
is equivalent to an additive $2^{k}-$quasi-cyclic code $C^{\prime }$ in ${
\mathbb{F}}_{2}^{n2^k}$.
\end{enumerate}
\end{theorem}

\begin{proof}
Let $C$ be a skew-cyclic code  in $A_k^{n}$.
 Consider  the case when $n$ is odd, where the
equation in Lemma~\ref{2and2'} holds.
Define the  permutation $\sigma_{S_1}$ by
\begin{equation}\label{eq:3}
 \sigma_{S_1} =
\begin{pmatrix}
%\left(
%    \begin{array}{cccccccccc}
      1 &  \ldots & 2^{k-1} & 2^{k-1}+1& \ldots & 2^k &\ldots& (2n-1)2^{k-1}+1&\ldots& n2^k\\
      \sigma_S^{2n-1}(1) & \ldots&\sigma_S^{2n-1}(2^{k-1})& \sigma_S^{2n-2}(1) & \ldots & \sigma_S^{2n-2}(2^{k-1})& \ldots&\sigma_S^0(1) &  \ldots &\sigma_S^0(2^{k}) \\
%    \end{array}
% \right).
\end{pmatrix}
\end{equation}
It is clear that for all $1 \leq j \leq n2^k$ such that $j= a 2 ^{k-1} +b$ with $0\leq a\leq 2n-1$ and $1\leq b \leq  2 ^{k-1}$
we have
\begin{equation}\label{eq:4}
 \sigma_{S_1}(j)=\sigma_S^{2n-1-a}(b).
\end{equation}
The permutation $\sigma_{S_1}$ induces the permutation $\Sigma_{S_1}$ acting
on the elements of $\F_{2}^{n2^k}$. For $\v=(v_1,v_2,\ldots,v_{n2^k}) \in \F_{2}^{n2^k}$,
\begin{equation}\label{eq:induce2}
{\Sigma}_{S_1}(\v)=\left(v_{\sigma_{S_1}(1)},v_{\sigma_{S_1}(2)},\ldots,v_{\sigma_{S_1}(n2^k)}\right)\text{.}
\end{equation}
To show that ${\Sigma}_{S_1}(\Phi_k(C))$ is a $2^{k-1}-$quasi-cyclic code,  we must prove that for all codewords
$\v \in \Phi_k(C)$,  we have that
$$T^{2^{k-1}}({\Sigma}_{S_1}(\v))\in {\Sigma}_{S_1}(\Phi_k(C))$$
where $T$ is the vector cyclic
shift.

Since $\Sigma_S(\Phi_k(C))=\Phi_k(C)$ by Lemma \ref{lem1}, we only need to show that
\begin{equation}\label{eq:5}
 T^{2^{k-1}}({\Sigma}_{S_1}(\v))={\Sigma}_{S_1}(\Sigma_S(\v)).
\end{equation}
We start with the right hand side of the equation. By definition,
\begin{equation}\label{eq:6a}
\Sigma_S(\v) = \left(v_{\sigma_S(1)},v_{\sigma_S(2)},\ldots,v_{\sigma_S (2n)}\right)
:= (v'_1,v'_2,\ldots,v'_{n2^k})=\v'.
\end{equation}

Applying $\sigma_{S_1}$ and by Equation (\ref{eq:3}), we have
\begin{align*}
{\Sigma}_{S_1}(\Sigma_S(\v)) &= \left(v'_{\sigma_{S_1}(1)},v'_{\sigma_{S_1}(2)},
\ldots,v'_{\sigma_{S_1}(2^{k-1})},v'_{\sigma_{S_1} (2^{k-1}+1)},
\ldots,v'_{\sigma_{S_1}(2^{k-1})},v'_{\sigma_{S_1} ((2n-1)2^{k-1}+1)},\ldots,v'_{\sigma_{S_1} (n2^{k})}\right)\\
&= \left(v'_{\sigma_S^{2n-1}(1)},v'_{\sigma_S^{2n-1}(2)},\ldots,v'_{\sigma_S^{2n-1}(2^{k-1})},v'_{\sigma_S^{2n-2}(1)},
\ldots,v'_{\sigma_S^{0}(1)},\ldots,v'_{\sigma_S^{0}(2^{k-1})}\right).
\end{align*}
 Now, Equation~(\ref{eq:6a}) allows us to write
 \begin{eqnarray*}
  {\Sigma}_{S_1}(\Sigma_S(\v) ) &=&\left(v_{\sigma_S^{2n}(1)},
  v_{\sigma_S^{2n}(2)},\ldots,v_{\sigma_S^{2n}(2^{k-1})},v_{\sigma_S^{2n-1}(1)},
\ldots,v_{\sigma_S(1)},\ldots,v_{\sigma_S(2^{k-1})}\right) \\
    &=& \left(v_{1},v_{2},\ldots,v_{2^{k-1}},v_{\sigma_S^{2n-1}(1)},
\ldots,v_{\sigma_S^{2n-1}(2^{k-1})},\ldots, v_{\sigma_S(1)},\ldots,v_{\sigma_S(2^{k-1})}\right)
 \end{eqnarray*}
by Lemma~\ref{lem1}.
Since
\begin{align*}
{\Sigma}_{S_1}(\v) & =  \left(v_{\sigma_{S_1}(1)},v_{\sigma_{S_1}(2)},\ldots,
v_{\sigma_{S_1}(2^{k-1})},v_{\sigma_{S_1}(2^{k-1}+1)},\ldots,v_{\sigma_{S_1}(2^{k})},
\ldots,v_{\sigma_{S_1} ((2n-1) 2^{k-1}+1)},\ldots,v_{\sigma_{S_1}(n2^{k})}\right)\\
 &= \left(v_{\sigma_S^{2n-1}(1)},v_{\sigma_S^{2n-1}(2)},\ldots,v_{\sigma_S^{2n-1}(2^{k-1})},v_{\sigma_S^{2n-2}(1)},\ldots,
 v_{\sigma_S^{2n-2}(2^{k-1})},\ldots,v_{\sigma_S^{0}(1)}
 ,\ldots,v_{\sigma_S^{0}(2^{k-1})}\right)\\
 &= \left(v_{\sigma_S^{2n-1}(1)},v_{\sigma_S^{2n-1}(2)},\ldots,v_{\sigma_S^{2n-1}(2^{k-1})},v_{\sigma_S^{2n-2}(1)},\ldots,
 v_{\sigma_S^{2n-2}(2^{k-1})},\ldots,v_{1}
 ,\ldots,v_{2^{k-1}}\right)
\end{align*}
and $\Sigma_S(\Phi_{k}(C))=\Phi_{k}(C)$, we get
\begin{equation*}
T^{2^{k}-1} ({\Sigma}_{S_1}(\v)) = {\Sigma}_{S_1}(\Sigma_S(\v)) \in {\Sigma}_{S_1}(\Phi_{k}(C))\text{.}
\end{equation*}
Thus, ${\Sigma}_{S_1}(\Phi_{k}(C))$ is a $2^{k-1}-$quasi-cyclic code.\\

Now let $n$ be an even integer.
Define the  permutation $\sigma_{S_2}$ by
%\scriptscriptstyle

\begin{equation}\label{eq:9}
 \sigma_{S_2} =
\left(
    \begin{array}{ccccccccccc}
      1 & 2 & \ldots & 2^{k} & 2^{k}+1& \ldots & 2^{k+1} &\ldots& (n-1)2^k +1&\ldots& n2^k\\
      \sigma_S^{n-1}(1) & \sigma_S^{n-1}(2) & \ldots&\sigma_S^{n-1}(2^{k})& \sigma_S^{n-2}(1)
      & \ldots & \sigma_S^{n-2}(2^{k})& \ldots& \sigma_S^0(1) &  \ldots &\sigma_S^0(2^{k}) \\
    \end{array}
 \right).
\end{equation}
Let us denote by  ${\Sigma}_{S_2}$ the permutation induced by  $ \sigma_{S_2} $ acting
on the elements of $\F_{2}^{n2^k}.$  For $\v=(v_1,v_2,\ldots,v_{2n}) \in \F_{2}^{n2^k}$,
\begin{equation}\label{eq:induce2}
{\Sigma}_{S_2}(\v)=\left(v_{\sigma_{S_2}(1)},v_{\sigma_{S_2}(2)},\ldots,v_{\sigma_{S_2}(n2^k)}\right).
\end{equation}
It is clear that for all $1 \leq j \leq n2^k$, $j= a 2 ^{k} +b$ with $0\leq a\leq n-1$ and
 $1\leq b \leq  2 ^{k}$
we have that
\begin{equation}\label{eq:4}
 \sigma_{S_2}(j)=\sigma_S^{n-1-a}(b).
\end{equation}

To show that ${\Sigma}_{S_2}(\Phi_k(C))$ is a $2^{k}-$quasi-cyclic code  we must prove that for all codewords
$\v \in \Phi_k(C)$, $$T^{2^{k}}({\Sigma}_{S_2}(\v))\in {\Sigma}_{S_2}(\Phi_k(C))$$ where $T$ is the vector cyclic
shift. Since $\Sigma_S(\Phi_k(C))=\Phi_k(C)$ by Lemma \ref{lem0}, we only need to show that
\begin{equation}\label{eq:5}
 T^{2^{k}}({\Sigma}_{S_2}(\v))={\Sigma}_{S_2}(\Sigma_S(\v))\text{.}
\end{equation}
We have that
\begin{equation}\label{eq:6}
\Sigma_S(\v) = \left(v_{\sigma_S(1)},v_{\sigma_S(2)},\ldots,v_{\sigma_S (2n)}\right)
:= (v'_1,v'_2,\ldots,v'_{n2^k})=\v'\text{.}
\end{equation}

Applying $\sigma_{S_2}$ and by Equation \ref{eq:induce2}, we have
\begin{align*}
{\Sigma}_{S_2}(\Sigma_S(\v)) &= \left(v'_{\sigma_{S_2}(1)},v'_{\sigma_{S_2}(2)},
\ldots,v'_{\sigma_{S_2}(2^{k})},v'_{\sigma_{S_2} (2^{k}+1)},\ldots,
v'_{\sigma_{S_2}(2^{k})},v'_{\sigma_{S_2} ((n-1)2^{k}+1)},\ldots,v'_{\sigma_{S_2} (n2^{k})}\right)\\
&= \left(v'_{\sigma_S^{n-1}(1)},v'_{\sigma_S^{n-1}(2)},\ldots,v'_{\sigma_S^{n-1}(2^{k})},v'_{\sigma_S^{n-2}(1)},
\ldots,v'_{\sigma_S^{0}(1)},\ldots,v'_{\sigma_S^{0}(2^{k})}\right)
\end{align*}
 Equation(\ref{eq:6}) allows us to write
 \begin{eqnarray*}
  {\Sigma}_{S_2}(\Sigma_S(\v) ) &=&\left(v_{\sigma_S^{n}(1)},v_{\sigma_S^{n}(2)},\ldots,v_{\sigma_S^{n}(2^{k})},v_{\sigma_S^{n-1}(1)},
\ldots,v_{\sigma_S(1)},\ldots,v_{\sigma_S(2^{k})}\right) \\
    &=& \left(v_{1},v_{2},\ldots,v_{2^{k}},v_{\sigma_S^{n-1}(1)},
\ldots,v_{\sigma_S^{n-1}(2^{k})},\ldots, v_{\sigma_S(1)},\ldots,v_{\sigma_S(2^{k})}\right)
 \end{eqnarray*}
by Lemma \ref{lem1}.
Since
\begin{align*}
{\Sigma}_{S_2}(\v) & =  \left(v_{\sigma_{S_2}(1)},v_{\sigma_{S_2}(2)},
\ldots,v_{\sigma_{S_2}(2^{k})},v_{\sigma_{S_2}(2^{k}+1)},\ldots,v_{\sigma_{S_2}(2.2^{k})}
\ldots,v_{\sigma_{S_2} ((n-1) 2^{k}+1)},\ldots,v_{\sigma_{S_2}(n2^{k})}\right)\\
 &= \left(v_{\sigma_S^{n-1}(1)},v_{\sigma_S^{n-1}(2)},\ldots,v_{\sigma_S^{n-1}(2^{k})},v_{\sigma_S^{n-2}(1)},\ldots,
 v_{\sigma_S^{n-2}(2^{k})},\ldots,v_{\sigma_S^{0}(1)}
 ,\ldots,v_{\sigma_S^{0}(2^{k})}\right)\\
 &= \left(v_{\sigma_S^{n-1}(1)},v_{\sigma_S^{n-1}(2)},\ldots,v_{\sigma_S^{n-1}(2^{k})},v_{\sigma_S^{n-2}(1)},\ldots,
 v_{\sigma_S^{n-2}(2^{k})},\ldots,v_{1}
 ,\ldots,v_{2^{k}}\right)
\end{align*}
and $\Sigma_S(\Phi_{k}(C))=\Phi_{k}(C)$, we get
\begin{equation*}
T^{2^{k}} ({\Sigma}_{S_2}(\v)) = {\Sigma}_{S_2}(\Sigma_S(\v)) \in {\Sigma}_{S_2}(\Phi_{k}(C))\text{.}
\end{equation*}
Thus, ${\Sigma}_{S_2}(\Phi_{k}(C))$ is a $2^{k}-$quasi-cyclic code. This completes the proof.
\end{proof}

We provide below two examples of the first characterization of $\Theta_S-$cyclic codes.

\begin{ex}
Let $C_1=\{0000,1010,1000,0010\}.$  We can see that $C$ is a binary $2-$quasi-cyclic code of length $4.$
Now, consider
\begin{itemize}
\item $\Phi_1^{-1}(0000)=(0,0),$
\item $\Phi_1^{-1}(1010)=(\phi_1^{-1}(10),\phi_1^{-1}(10))=(1+v,1+v),$
\item $\Phi_1^{-1}(1000)=(\phi_1^{-1}(10),\phi_1^{-1}(00))=(1+v,0),$
\item $\Phi_1^{-1}(0010)=(\phi_1^{-1}(00),\phi_1^{-1}(10))=(0,1+v).$
\end{itemize}

Then we have $\Phi_1^{-1}(C_1)=\{(0,0),(1+v,1+v),(1+v,0),(0,1+v)\}$ which is a $\Theta_\emptyset-$cyclic code over $A_1.$

\end{ex}

\begin{ex}
Let $C_2=\langle (1,1,1)\rangle$ be a $\Theta_{\{1,2\}}-$cyclic code of length 3 over $A_2.$ Consider,
\begin{itemize}
\item $\Phi_2(111)=(\phi_1(1),\phi_1(1),\phi_1(1),\phi_1(1),\phi_1(1),\phi_1(1))=111111111111,$
\item $\Phi_2(v_1v_1v_1)=(\phi_1(v_1),\phi_1(v_1),\phi_1(v_1),\phi_1(v_1),\phi_1(v_1),\phi_1(v_1))=010101010101,$
\item $\Phi_2(v_2,v_2,v_2)=(\phi_1(0),\phi_1(1),\phi_1(0),\phi_1(1),\phi_1(0),\phi_1(1))=001100110011,$
\item $\Phi_2(v_1v_2,v_1v_2,v_1v_2)=(\phi_1(0),\phi_1(v_1),\phi_1(0),\phi_1(v_1),\phi_1(0),\phi_1(v_1))=000100010001.$
\end{itemize}

We can see that, $\Phi_2(C_2)=\langle 111111111111,010101010101,001100110011,000100010001\rangle$
which is a binary $2-$quasi-cyclic code.
\end{ex}

\subsection{Second characterization}

We also have a characterization of $\Theta_S-$cyclic codes using the map $\Psi_{k,p}$ as follows.

\begin{theorem}
An $A_k$-linear code $C$ is $\Theta_S-$cyclic of length $n$ if and only if
$C=\Psi_{k,p}^{-1}(C_1,\dots,C_s)$, where $C_1,\dots,C_s$ are quasi-cyclic codes of index $2$ which satisfy
\begin{equation}
T_{\Theta_{S'}}(C_i)\subseteq C_{\mu(i)}
\label{mu}
\end{equation}
for some $S'\subseteq S$ and permutation $\mu$.
\label{construct}
\end{theorem}
\begin{proof}
Let $S=\{i_1,\dots,i_t\}$, then take $p=i_j-1$ for some $i_j\in S$.
We will have codes $C_1,\dots,C_s$ over $A_p$  such that
$C=\Psi_{k,p}^{-1}(C_1,\dots,C_s)$. It is clear that $C_1,\dots,C_s$ are quasi-cyclic codes of
index $2$ since $\Theta_S$ is an involution. Now,
for any $i$, let $c_i\in C_i$, and write
\[
c_i=\left(\sum_{D\subseteq B_i}\alpha_D^{(1)},\dots,\sum_{D\subseteq B_i}\alpha_D^{(n)}\right).
\]
Consider,
\[
c_i'=\Psi_{k,p}(0,\dots,0,c_i,0,\dots,0)=\left(\left(\sum_{D\subseteq B_i}\alpha_D^{(1)}\right)
\sum_{B\subseteq \Omega_p,B\supseteq B_i}w_B
,\dots,\left(\sum_{D\subseteq B_i}\alpha_D^{(n)}\right)\sum_{B\subseteq \Omega_p,B\supseteq B_i}w_B\right)\in C
\]
and
\[
T_{\Theta_S}(c_i')=\left(\left(\sum_{D\subseteq B_i}\Theta_{S'}(\alpha_D^{(n)})\right)
\sum_{B\subseteq \Omega_p,B\supseteq B_i}\Theta_S(w_B)
,\dots,\left(\sum_{D\subseteq B_i}\Theta_{S'}(\alpha_D^{(n-1)})\right)\sum_{B\subseteq \Omega_p,B\supseteq B_i}
\Theta_S(w_B)\right)
\]
where $S'=\{i_1,\dots,i_{j-1}\}\subseteq S$. When we consider $\Psi_{k,p}(\Theta_S(w_B))$,
we can think of it is a permutation of $\Psi_{k,p}(w_B)$
as in Equation~(\ref{olfa}), which gives permutation $\mu$. Therefore, we have
\[\Psi_{k,p}(T_{\Theta_S}(c_i'))=(0,\dots,0,T_{\Theta_{S'}}(c_i),0,\dots,0)\]
where $T_{\Theta_{S'}}(c_i)\in C_{\mu(i)}$, as we hope.

 Using the above setting for $p$, $S$, and $S'$, for any $c\in C$, let
\[
c=\left(\sum_{B\subseteq \Omega_p}\alpha_{B}^{(1)},\dots,\sum_{B\subseteq \Omega_p}\alpha_{B}^{(n)}\right).
\]
By Equation~(\ref{olfa}), we have
\[\Psi_{k,p}(\Theta_S(c))=\left(\begin{array}{c}
\sum_{D\subseteq B_{\mu(1)}}\Theta_{S'}(\alpha_D^{(1)}),\dots,
\sum_{D\subseteq B_{\mu(s)}}\Theta_{S'}(\alpha_D^{(1)})\\
\vdots \\
\sum_{D\subseteq B_{\mu(1)}}\Theta_{S'}(\alpha_D^{(n)}),\dots,
\sum_{D\subseteq B_{\mu(s)}}\Theta_{S'}(\alpha_D^{(n)})
\end{array}
\right),\]
which is in $(C_1,\dots,C_s)$ by assumption.
\end{proof}

Two examples for the second characterization of $\Theta_S-$cyclic codes are given below.

\begin{ex}
Let $C_3=\langle (v_2,0), (0,1+v_2), (v_2,1+v_2)\rangle.$ We can check that $C_3$ is
a $\Theta_{\{1,2\}}-$cyclic code of length 2 over $A_2.$ Consider,
\begin{itemize}
\item ${\Psi}_{2,1}(v_2,0)=(0,1,0,0),$
\item ${\Psi}_{2,1}(0,1+v_2)=(0,0,1,0),$
\item ${\Psi}_{2,1}(v_2,1+v_2)=(0,1,1,0),$
\item ${\Psi}_{2,1}(v_1v_2,0)=(0,v_1,0,0),$
\item ${\Psi}_{2,1}(0,v_1+v_1v_2)=(0,0,v_1,0),$
\item ${\Psi}_{2,1}(v_1v_2,v_1+v_1v_2)=(0,v_1,v_1,0).$
\end{itemize}

We have that $C_3={\Psi}_{2,1}^{-1}(C_4,C_5),$ where $C_4=\langle (0,1)\rangle$ and
$C_5=\langle (1,0)\rangle$ which are $2-$quasi-cyclic codes over $A_1.$
Note that, $\mu(1)=2$ and $\mu(2)=1.$ Finally, we can see that
$T_{\Theta_{\{1\}}}(C_1)=C_2=C_{\mu(1)}$ and $T_{\Theta_{\{1\}}}(C_2)=C_1=C_{\mu(2)}.$
\end{ex}

\begin{ex}
Let $C_6=C_7=\{00,11\}$ are binary cyclic codes. Consider,
\begin{itemize}
\item ${\Psi}_{1,0}^{-1}(00,11)=(\Psi_{1,0}^{-1}(01),\Psi_{1,0}^{-1}(01))=(v,v),$
\item ${\Psi}_{1,0}^{-1}(11,00)=(\Psi_{1,0}^{-1}(10),\Psi_{1,0}^{-1}(10))=(1+v,1+v),$
\item ${\Psi}_{1,0}^{-1}(11,11)=(\Psi_{1,0}^{-1}(11),\Psi_{1,0}^{-1}(11))=(1,1).$
\end{itemize}

So, we have ${\Psi}_{1,0}^{-1}(C_6,C_7)=\{(0,0),(v,v),(1+v,1+v),(1,1)\}$
which is a $\Theta_{\{1\}}-$cyclic code over $A_1.$
\label{exprev}
\end{ex}

\section{Construction of $\Theta_S-$cyclic codes over $A_k$}\label{algorithm}

In this section, we illustrate some constructions of skew cyclic codes over $A_k.$

Recall that the ring $A_k$ is a principal ideal ring.  As such it is isomorphic to a direct product of chain rings.
In particular,
$A_k$ is isomorphic as a ring to  $\FF_2^{2^k}$ via the Chinese Remainder Theorem,
see \cite{C-D-D} for a complete description of this.
Let $CRT: \FF_2^{2^k} \rightarrow  A_k$ be this canonical map.  Let $CRT  (C_1, ....,
C_{2^{k}})$ be the code over $A_k$ formed by taking the map from $C_1 \times C_2 \times \dots \times C_{2^k}$
where each $C_i$ is a binary code.

Define the following map
$\Gamma:{\mathbb{F}}_2^{n2^k} \longmapsto {\mathbb{F}}_2^{n}\times {\mathbb{F}}
_2^{n} \times \cdots \times {\mathbb{F}}_2^{n} $
by
\[
\Gamma( x_{1}^{1},\ldots, x_{1}^{2^k}, x_{2}^{1},\ldots,x_{1}^{2^k},
x_{3}^{1},\ldots, x_{n}^{1},\ldots, x_{n}^{2^k})  =
((x_{1}^{1},\ldots, x_{n}^{1}), (x_{1}^{2},\ldots, x_{n}^{2}), \ldots,
(x_{1}^{2^k},\ldots, x_{n}^{2^k})).
\]

For all codes $C$ over $A_k$ with $C = CRT (C_1, \ldots,C_{2^{k}})$ we have that
\[
\Gamma\circ \Phi_k (C) =(C_1,C_2, \ldots, C_{2^{k}}).
\]

\begin{prop}
Let $n$ be an even integer  and  let $C_1,C_2, \ldots, C_{2^{k}}$ be binary cyclic codes in
${\mathbb{F}}_2^{n}.$  Then for all $\ A\subseteq  \{1,2,\dots, k \}$   there  exists a $\Theta_{A}-$cyclic code $C$
in $A_k^{n}.$
\end{prop}

\begin{proof}
It is clear that $C^\prime= \Gamma^{-1}(C_1,C_2, \ldots, C_{2^{k}})$ is a
$2^{k}-$quasi-cyclic code over $\F_2^{n2^k}$.  It follows from
Theorem~\ref{th1} that
$\Phi_{k}^{-1}\circ \sigma_{S_2} ^{-1}( C^\prime)$ is a $\Theta_{A}-$cyclic code $C$  in $A_k^{n}.$
\end{proof}

We define the map
\begin{equation}
\Gamma_1: {\mathbb{F}}_2^{n2^{k-1}} \longmapsto {\mathbb{F}}_2^{2n}\times
{\mathbb{F}}_2^{2n}\times\cdots \times{\mathbb{F}}_2^{2n},
\end{equation}
with
\begin{multline*}
\Gamma_1( x_{1}^{1},\ldots, x_{1}^{2^{k-1}}, x_{2}^{1}, \ldots,x_{2}^{2^{k-1}},\ldots,
x_{2n}^{1},\ldots, x_{2n}^{2^{k-1}}) \\=( (x_{1}^{1},\ldots,
x_{2n}^{1}), (x_{1}^{2}, \ldots,x_{2n}^{2}), \ldots,
(x_{1}^{2^{k-1}},\ldots, x_{2n}^{2^{k-1}})).
\end{multline*}

\begin{prop}
Let $n$ be an odd integer and let $C_1, C_2,\ldots, C_{2^{k-1}}$ be binary cyclic codes in
${\mathbb{F}}_2^{2n}.$  Then for all $\ S \subseteq  \{1, 2, \ldots, k \}  $ there  exists a $\Theta_{S}-$cyclic code $C$ in
$A_k^{n}$.
\end{prop}
\begin{proof}
It is clear that $C^\prime= \Gamma_1^{-1}(C_1,\ldots, C_{2^{k-1}})$ is
a $2^{k-1}-$quasi-cyclic code over $\F_2^{n2^k}$. It follows from Theorem~\ref{th1} that
$\Phi_{k}^{-1}\circ \sigma_{S_1} ^{-1}( C^\prime)$ is a $\Theta_{S}-$cyclic code $C$  in $A_k^{n}.$
\end{proof}

We know describe an algorithm for constructing $\Theta_S-$cyclic codes.

\begin{enumerate}
\item {\bf Construction of $\Theta_{A}-$cyclic codes in $A_k$ of even length}

\begin{enumerate}
\item We consider $C_1, \ldots, C_{2^{k}}$ binary cyclic codes in ${\mathbb{F}}_2^{n}$

\item We apply $\Gamma^{-1}$ and we obtain $C^{\prime }$ a $2^{k}-$
quasi-cyclic code in ${\mathbb{F}}_2^{n2^k}$.

\item We apply $\Phi_{k}^{-1}\circ \sigma_{S_2} ^{-1}$ to $C^{\prime }$. We
obtain a $\Theta_{S}-$cyclic code $C$ in $A_k$.
\end{enumerate}

\item {\bf Construction of $\Theta_S-$cyclic codes in $A_k$ of odd length }

\begin{enumerate}
\item We consider $C_1, \ldots, C_{2^{k-1}}$ binary cyclic codes in ${\mathbb{F}}_2^{2n}$

\item We apply $\Gamma_1^{-1}$ and we obtain $C^{\prime }$ a $2^{k-1}-$
quasi-cyclic code in ${\mathbb{F}}_2^{n2^k}$.

\item We apply $\Phi_{k}^{-1}\circ \sigma_{S_1}^{-1}$ to $C^{\prime }$. We
obtain a $\Theta_{S}-$cyclic code $C$ in $A_k^{n}$.
\end{enumerate}

\item {\bf Construction of $\Theta_S-$cyclic codes over $A_k$ from codes over $A_p$, where $p<k$ }

\begin{enumerate}
\item Given $C_1,\dots, C_s$ quasi-cyclic codes of index $2$ in $A_p$
which satisfy Equation~\ref{mu} in Theorem~\ref{construct}, for
some $S^\prime\subseteq S$.
\item Appling $\Psi_{k,p}$ to $(C_1,\ldots,C_s)$, we obtain a $\Theta_S-$cyclic code over $A_k$.
\end{enumerate}

\end{enumerate}

In terms of skew-polynomial rings, the third construction of a $\Theta_S-$cyclic code above will be as follows.
\begin{prop}
Let $C=\Psi_{k,p}^{-1}(C_1,\dots,C_s)$ be $\Theta_S-$cyclic codes over $A_k$, where $C_1,\dots,C_s$ are codes over $A_{p}$, for some $p< k$.
If $C_i=\langle g_{1_i}(x),\dots,g_{m_i}(x)\rangle$, for all $i=1,\dots,s$, then
\[
C=\langle g_{1_1}(x),\dots,g_{m_1}(x),\ldots,g_{1_s}(x),\dots,g_{m_s}(x)\rangle.
\]
\label{generator}
\end{prop}
\begin{proof}
For any $c(x)\in C$, there exist $c_i(x)\in C_i$, where $1\leq i\leq s$, such that $c(x)=\Psi_{k,p}^{-1}(c_1(x),\dots,c_s(x))$. Now, let
\[
c_j(x)=\sum_{k=1}^{m_j}\alpha_{kj}(x)g_{jk}
\]
for some $\alpha_{kj}(x)\in A_p[x,\Theta_S^\prime]$, for all $j=1,\dots,s$. Then we have
\[
\begin{aligned}
 c(x)  &=  w_{B_1}\left(\sum_{k=1}^{m_1}\alpha_{k1}(x)g_{1k}\right)\\
         & \quad +\dots+w_{B_i}\left(\sum_{k=1}^{m_i}\alpha_{ki}(x)g_{ik} -
           \sum_{B_j\subseteq B_i}\left(\sum_{k=1}^{m_j}\alpha_{kj}(x)g_{jk}\right)\right) \\
         & \quad +\dots+w_{B_s}\left(\sum_{k=1}^{m_s}\alpha_{ks}(x)g_{sk}-\sum_{j=1}^s
           \left(\sum_{k=1}^{m_j}\alpha_{kj}(x)g_{jk}\right)\right)
\end{aligned}
\]
as we hope.
\end{proof}

\section{Self-dual codes}
In this section we will give a characterization for Hermitian self-dual codes over $A_k$.

First we will show that the orthogonal of a $\Theta_S-$cyclic code is again a $\Theta_S-$cyclic code.
We note that the Euclidean inner-product is simply the Hermitian inner-product with $R = \emptyset.$
Hence any proof for the Hermitian inner-product applies as well to the Euclidean inner-product.

\begin{theorem}
If C is a $\Theta_S-$cyclic code then $C^{H_R}$ is a $\Theta_S-$cyclic code for all $R \subseteq \{1,2,\dots,k\}$.
\end{theorem}
\begin{proof}
Let $S,R \subseteq  \{1,2,\dots,k\}.$  Let $\vu \in C$ and $\vw \in C^{H_R}.$
We know that since $\vu$ is an element of $C$ that $T^i_{\Theta_S}(\vu) \in C$ for all $i.$  Hence we have that
$[ T^i_{\Theta_S}(\vu), \vw]_{H_R} = 0.$

Then consider $\Theta_S(\Theta_R( [ T^i_{\Theta_S}(\vu), \vw]_{H_R} ) )=0.$

\begin{align*}
%\begin{eqnarray*}
\Theta_S(\Theta_R( [ T^i_{\Theta_S}(\vu), \vw]_{H_R}) )
&= \Theta_S\left(\Theta_R \left(\sum T^i_{\Theta_S} (u_j) \Theta_R(w_j) \right)\right) \\
&= \Theta_S\left(\Theta_R\left( \sum (\Theta_S(u_j) T^i (\Theta_R(w_{j-i})\right)\right) \\
&= \sum \Theta_R(u_j) T^i \Theta_S(w_{j-i}) \\
&= 0.
%\end{eqnarray*}
\end{align*}
This gives that $T^i_{\Theta_S}(\vw) \in C^{H_R}$ and so $C^{H_R}$ is $\Theta_S-$cyclic.
\end{proof}

Next, we have the following lemma.

\begin{lem}
A code $C$ over $A_k$ is self-dual if and only if $C=CRT(C_1,\dots,C_{2^k})$ and $C_i=C_i^\bot$ for all i.
\end{lem}
\begin{proof}
 Follows from Theorem 6.4 in \cite{Dough3}.
\end{proof}

Then we also have the following characterization.

\begin{theorem}\label{poiuy}
A code $C$ over $A_k$ satisfies $C=C^{H_T}$ if and only if $C=CRT(C_1,\dots,C_{2^k})$ and $C_i=C_{\mu(i)}$,
for some permutation $\mu$, for all $i$.
\end{theorem}
\begin{proof}
Recall that $C^{H_T}=(\Theta_T(C))^\bot$. By Equation~\ref{olfa}, we can say that the map $\Theta_T$
induces a permutation $\mu$ on the coordinates
of $\Phi_k(A_k)$. So, if
$C=CRT(C_1,\dots,C_{2^k})$ then
\[\Theta_T(C)=CRT(C_{\mu(1)},\dots,C_{\mu(2^k)}),\]
and the result follows.
\end{proof}

As an easy consequence, we have the following corollary.

\begin{cor}
A $\Theta_S-$cyclic code  $C$ is a Hermitian self-dual code with respect to $\Theta_T$ if and only if
$C=CRT(C_1,\dots,C_{2^k})$,   $C_i=C_{\mu(i)}$ for all $i$, where $\mu$ is a permutation induced by $\Theta_T$,
and $\Phi_k(C)$ is fixed by the permutation $\Sigma_S.$
\end{cor}
\begin{proof}
Follows from  Lemma~\ref{lem0} and Theorem~\ref{poiuy}.
\end{proof}

Here is an example of a self-dual code.

\begin{ex}
Let $C_8=\{(0,0),(v,v),(1+v,1+v),(1,1)\}$ is a $\Theta_{\{1\}}-$cyclic code over $A_1.$
Furthermore, if $T=\{1\},$ then $C_8^{H_T}=C_8.$
Using the calculation in Example~\ref{exprev},
we have that $C_8=\overline{\Psi}_{1,0}^{-1}(C,C'),$ where $C=C'=\{00,11\}$
which are two binary cyclic self-dual codes.
\end{ex}

\section{Optimal codes}

In this section, we will give a way to construct optimal $\Theta_S-$cyclic
codes using Proposition~\ref{generator} and some examples of optimal codes obtained using this technique.
First, we need the following lemma.

\begin{lem}
Let $C$ be a linear code over $A_k$. Then,
if $C=\Psi_{k,p}^{-1}(C_1,\dots,C_s)$, for some codes $C_1,\dots,C_s$  over $A_{p}$, for some $p$, then
$d_H(C)=\min_{1\leq i\leq s} d_H(C_i)$.
\label{jarak}
\end{lem}

\begin{proof}
 Follows from Lemma 6.2 in \cite{Dough3}.
\end{proof}

This means that we can use the optimal binary codes or optimal $\Theta_{S'}-$cyclic codes over $A_p$
to construct  optimal $\Theta_{S}-$cyclic
codes over $A_k$ for all $k$ with respect to the Hamming weight. There are tables for binary optimal Euclidean self-dual
codes in \cite{con} and for optimal Hermitian
self-dual codes and Hermitian Type IV self-dual codes in \cite{betsumiya-harada}. Therefore, in particular,
we only need the optimal binary quasi-cyclic codes to construct (self-dual) optimal $\Theta_S-$cyclic
codes over $A_k$ using the result in Proposition \ref{generator} or the algorithm described in Section \ref{algorithm}. For example,
we can use $\Theta-$cyclic codes over $A_1$ in
\cite{Nuh} to produce the optimal $\Theta_S-$cyclic codes over $A_k$ for $k\geq 2$, as in Table \ref{tableoptimal}.

\begin{longtable}{|c|c|c|c|c|c|}\hline\hline  \label{tableoptimal}
%\begin{tabular}{|c|c|c|c|c|c|}\hline\hline
$n$ & $d$ & Generator polynomial & $A_k$ & S & T\\[2mm]\hline\hline
4 & 2 & $x^2+1$ & $A_2$ & $\{1,2\}$  & $\{1\}$\\ \hline
6 & 2 & $x^3+1$ & $A_7$ & $\{1,3,5,7\}$  & $\{1\}$\\ \hline
8 & 4 & $x^4+(v_1+1)x^3+x^2+v_1x+1$ & $A_2$ & $\{1,2\}$  & $\emptyset$\\ \hline
10 & 2 & $x^5+1$ & $A_3$ & $\{1,2\}$  & $\emptyset$ \\ \hline
12 & 4 & $x^6+(v_1+1)x^5+x^4+x^3+x^2+v_1x+1$ & $A_3$ & $\{1,3\}$   & $\emptyset$ \\ \hline
14 & 4 & $x^7+x^6+x^5+x^4+x+1$ & $A_4$ & $\{1,2,3\}$ & $\{1\}$ \\ \hline
16 & 4 & $x^{8}+(v_1+1)x^{5}+x^4+v_1x^3+1$ & $A_4$ & $\{1,2,3\}$  & $\{1\}$ \\ \hline
18 & 4 & $x^9+x^7+v_1x^6+(1+v_1)x^5+(1+v_1)x^4+v_1x^3+x^2+1$ & $A_2$ & \{1,2\}  & $\{1\}$ \\ \hline
20 & 4 & $x^10+(v_1+1)x^7+x^6+x^5+x^4+v_1x^3+1$ & $A_6$ & $\{1,2,6\}$   & $\emptyset$ \\ \hline
22 & 6 &\begin{tabular}{l}
          $x^{11}+x^{10}+v_1x^9+v_1x^8+v_1x^7+(v_1+1)x^6$\\
	  $ +(v_1+1)x^5+v_1x^4+v_1x^3+v_1x^2+x+1$
         \end{tabular}
 & $A_3$ & $\{1,2\}$  & $\{1\}$ \\ \hline
24 & 8 & \begin{tabular}{l}
          $x^{12}+x^{11}+v_1x^{10}+x^9+(v_1+1)x^7+v_1x^5$\\
	  $ +x^3+(v_1+1)x^2+x+1$
         \end{tabular} & $A_5$ & $\{1,3,4,5\}$   & $\emptyset$\\ \hline
26 & 6 & \begin{tabular}{l}
          $x^{13}+x^{11}+v_1x^{10}+(1+v_1)x^9+v_1x^8+v_1x^7+(1+v_1)x^6$\\
	  $ +(1+v_1)x^5+v_1x^4+(1+v_1)x^3+x^2+1$
         \end{tabular} & $A_2$ & $\{1,2\}$  & $\{1\}$\\ \hline
%\end{tabular}\\
\caption{Table of examples of optimal self-dual $\Theta_S-$cyclic codes}
\end{longtable}

%\begin{longtable}{|c|c|c|c|c|}\hline \hline  \label{tableoptimal}
%%\begin{tabular}{|c|c|c|c|c|}\hline\hline
%$n$ & $d$ & Generator polynomial & $A_k$ & $S$\\[2mm]\hline\hline
%6 & 2 & $x^3+1$ & $A_7$ & $\{1,3,5,7\}$ \\ \hline
%8 & 4 & $x^4+(v_1+1)x^3+x^2+v_1x+1$ & $A_2$ & $\{1,2\}$\\ \hline
%16 & 4 & $x^{8}+(v_1+1)x^{5}+x^4+v_1x^3+1$ & $A_4$ & $\{1,2,3\}$ \\ \hline
%12 & 4 & $x^6+(v_1+1)x^5+x^4+x^3+x^2+v_1x+1$ & $A_3$ & $\{1,3\}$ \\ \hline
%20 & 4 & $x^10+(v_1+1)x^7+x^6+x^5+x^4+v_1x^3+1$ & $A_6$ & $\{1,2,6\}$ \\ \hline
%22 & 6 &\begin{tabular}{l}
%          $x^{11}+x^{10}+v_1x^9+v_1x^8+v_1x^7+(v_1+1)x^6$\\
%	  $ +(v_1+1)x^5+v_1x^4+v_1x^3+v_1x^2+x+1$
%         \end{tabular}
% & $A_3$ & $\{1,2\}$\\ \hline
%24 & 8 & \begin{tabular}{l}
%          $x^{12}+x^{11}+v_1x^{10}+x^9+(v_1+1)x^7+v_1x^5$\\
%	  $ +x^3+(v_1+1)x^2+x+1$
%         \end{tabular} & $A_5$ & $\{1,3,4,5\}$\\
%\hline
%%\end{tabular}\\
%\caption{Table of examples of optimal $\Theta_S-$cyclic codes}
%\end{longtable}

\section*{Acknowledgement}
A part of the work of D.S.  was done
when he visited Research Center for Pure and Applied Mathematics (RCPAM), Graduate School of
Information Sciences, Tohoku University Japan, on June 2015 under the financial support from
\emph{Riset Desentralisasi ITB-Dikti 2015} (Number 311c/I1.C01/PL/2015).
A.B., I., and I.M.-A. were supported in part by \emph{Riset dan Inovasi KK ITB tahun 2014}.


\begin{thebibliography}{}


\bibitem{Nuh} T. Abualrub, N. Aydin, P. Seneviratne, On $\Theta-$cyclic codes over $\FF_2 + v \FF_2$,
{\it Austrl. Journal of Comb.} {\bf 54}, 2012, 115 - 126.

\bibitem{betsumiya-harada}
K. Betsumiya, and M. Harada, Optimal Self-Dual Codes over $\mathbb{F}_2\times\mathbb{F}_2$ with Respect to The Hamming Weight,
{\it IEEE Trans. Inform. Theory} {\bf 50} (2), 2004, 356 - 358.

\bibitem{Boucher1} D. Boucher, W. Geiselmann and F. Ulmer, Skew-cyclic codes,
{\it Appl. Algebra in Eng., Commun. and Comp.} {\bf 18}, 2007, 379 - 389.

\bibitem{con}
J.H. Conway, and N.J.A Sloane, A New Upper Bound on the Minimal Distance of Self-Dual Codes, {\it IEEE Trans. Inform. Theory}
 {\bf 36} (6),  1990, 1319 - 1333.

\bibitem{C-D-D}
Y. Cengellenmis, A. Dertli, and S.T. Dougherty, Codes over an infinite family of rings with a Gray map,
{\it Des., Codes and Cryptog.} {\bf 72} (3),  2014, 559 - 580.

\bibitem{Dertli} A. Dertli,  Y. Cengellenmis,  MacDonald codes over the ring $\FF_2+v \FF_2$.
{\it Int. J. Algebra} {\bf 5}, 2011,   985 - 991.

\bibitem{TypeIV} S.T. Dougherty, P. Gaborit, M. Harada, A. Munemasa, and P. Sol\'e.
Type IV Self-Dual Codes over Rings , {\it IEEE Trans. Inform. Theory} {\bf  45} (7),  November 1999,  2345 - 2360.

\bibitem{Dough3} S.T. Dougherty, J.L. Kim, and H. Kulosman,
MDS Codes over Finite Principal Ideal Rings, {\it Des., Codes, and Cryptog.} {\bf 50} (1), 2009, 77 - 92.

\bibitem{Dough1} S.T. Dougherty, B. Yildiz, and S. Karadeniz, Codes over $R_k$,
Gray maps and their Binary Images, {\it Finite Fields and their Appl.} {\bf  17} (3), 2011, 205 - 219.

\bibitem{Dough2} S.T. Dougherty, B. Yildiz, and S. Karadeniz, Cyclic Codes over $R_k$,
{\it Des., Codes, and Cryptog.,} {\bf 63} (1), 2012, 113 - 126.

\bibitem{Dough2a} S.T. Dougherty, B. Yildiz, and S. Karadeniz, Self-dual Codes over $R_k$ and
Binary Self-Dual Codes,  {\it Eur. Journal of Pure and Appl. Math.} {\bf 6} (1), 2013, 89 - 106.

\bibitem{gao}
J. Gao, Skew Cyclic Codes over $\mathbb{F}_p+v\mathbb{F}_p$,
{\it J. Appl. Math. and Informatics} {\bf 31} (3-4), 2013, 337 - 342.

\bibitem{rim}
M. Rimmer, Isomorphism between Skew Polynomial Rings, {\it J. Austral. Math. Soc.}, {\bf 25} (Series A), 1978, 314 - 321.

\bibitem{Siap}
I. Siap, T. Abualrub, N. Aydin, and P. Seneviratne, Skew cyclic codes of arbitrary length,
{\it Int. J. Inform. Coding Theory}, {\bf 2} (1), 2011, 10 - 20.

\end{thebibliography}
\end{document}